\documentclass[lettersize,journal]{IEEEtran}
\usepackage{amsmath,amsfonts}
\usepackage{ amssymb, amsthm}
\usepackage{tikz}
\usetikzlibrary{positioning}
\usepackage[ruled,vlined]{algorithm2e}
\usepackage[font=footnotesize]{caption} 
\usepackage{wrapfig}
\usepackage{graphicx} 
\usepackage{tikz}
\usetikzlibrary{arrows}
\usepackage{caption}
\usepackage{array}

\newtheorem{remark}{Remark}
\newtheorem{example}{Example}
\usepackage[caption=false,font=normalsize,labelfont=sf,textfont=sf]{subfig}
\usepackage{textcomp}
\usepackage{stfloats}
\usepackage{url}
\usepackage{verbatim}
\usepackage{xcolor}
\usepackage{graphicx}
\usepackage{cite}
\newtheorem{theorem}{\textbf{Theorem}}
\newcommand{\RN}[1]{%
\textup{\uppercase\expandafter{\romannumeral#1}}%
}

\begin{document}

\title{Single-Server Pliable Private Information Retrieval\\ with Identifiable Side Information}

\author{Megha Rayer, Charul Rajput, B. Sundar Rajan\\ Department of Electrical Communication Engineering, Indian Institute of Science, Bengaluru 560012, KA, India\\ Emails: megharayer@iisc.ac.in, charul.rajput9@gmail.com, bsrajan@iisc.ac.in

}
\maketitle

\begin{abstract}
In Pliable Private Information Retrieval (PPIR) with a single server, messages are partitioned into $\Gamma$ non-overlapping classes. The user wants to retrieve a message from its desired class without revealing the identity of the desired class to the server. In [S. A. Obead,  H. Y. Lin and E. Rosnes, \textit{“Single-Server Pliable Private Information Retrieval With Side Information,” arXiv:2305.06857 [cs.IT]}], authors consider the problem of PPIR with Side Information (PPIR-SI), where the user now has side information. The user wants to retrieve any new message (not included in the side information) from its desired class without revealing the identity of the desired class and its side information. A scheme for the PPIR-SI is given by Obead et al. for the case when the user’s side information is unidentified, and this case is referred to as PPIR with Unidentifiable SI (PPIR-USI). In this paper, we study the problem of PPIR for the single server case when the side information is partially identifiable, and we term this case as PPIR with Identifiable Side Information (PPIR-ISI). The user is well aware of the identity of the side information belonging to $\eta$ number of classes, where $1\leq \eta \leq \Gamma$. In this problem, The user wants to retrieve a message from its desired class without revealing the identity of the desired class to the server. 
We give a scheme for PPIR-ISI, and we prove that having identifiable side information is advantageous by comparing the rate of the proposed scheme to the rate of the PPIR-USI scheme given by Obead et al. for some cases. Further, we extend the problem of PPIR-ISI for multi-user case, where users can collaborately generate the query sets, and we give a scheme for this problem.

\end{abstract}


\section{Introduction}
In the era of data-driven applications, preserving user privacy while providing efficient rates has become a critical challenge.
The rapid growth of online services and applications has led to the accumulation of vast amounts of user data, which, when accessed without proper safeguards, can lead to privacy breaches and unauthorized disclosure of sensitive information.
Nowadays, in content delivery networks, network operators have the ability to deduce users' preferences and content popularity, contributing to malicious practices \cite{ref1},\cite{ref2}.
The Pliable Private Information Retrieval (PPIR) problem was introduced in \cite{ref5} in which the user is interested in any message from a desired subset of the available dataset without revealing the identity of the desired subset to the server. The PPIR problem itself is a variant of the Private Information Retrieval (PIR) problem, which has emerged as a pivotal field of study, addressing the balance between data accessibility and users' privacy.

\subsection{Prior Work}
B. Chor et al. \cite{ref3} formulated the foundational work that introduced the concept of PIR and explored the possibility of retrieving information from a single server by the user without revealing the desired information's identity. Subsequent research expanded upon this work to improve communication complexity, introduce information-theoretic security, and explore multi-server PIR, achieving efficiency and robustness.
With the onset of complex data-sharing scenarios and the need for more flexible PIR schemes, PPIR has garnered attention. The concept of pliability in PIR schemes addresses the need for flexible data retrieval in scenarios where the user focuses on retrieving derived data rather than the exact data. Obead and Kliewer \cite{ref5} have introduced the concept of PPIR, where the server is divided into $\Gamma$ classes, and the user derives any new message from its desired class.

Content distribution networks consist of a server with a large number of messages that are partitioned into classes based on content types. The server is connected to the users via a broadcast medium. Users randomly join in and out of the network. In such scenarios, users obtain messages from the server apriori as the messages may have been broadcasted earlier. This information is known as the side information. This problem is modelled as single-server PPIR side information (PPIR-SI), and widely studied in \cite{ref4,ref7,ref8,ref9,ref10,ref11,ref12,ref13,ref14}. Obead et al. \cite{ref6} have briefly discussed the case where the user has minimum knowledge about the server and their side information. This case was termed as PPIR unidentified side information (PPIR-USI).

In this paper, we discuss the case when some side information with the user is identifiable, i.e., the PPIR-ISI problem, and propose a scheme using maximum distance separable (MDS) codes \cite{LX2004}. An $[n, k]$ MDS code has the property that any $k$ symbols out of $n$ symbols are information symbols, i.e., the information can be recovered from any set of $k$ code symbols.
The PPIR-USI Scheme \cite{ref6} also makes use of MDS codes in which encoding is done for the messages within a class $i, i \in [\Gamma]$, whereas, in the proposed PPIR-ISI Scheme, encoding is done for the messages across all the classes.

\subsection{Contributions and organization}
To the best of our knowledge, the PPIR-ISI problem has not been examined before. In this work, we formally define the PPIR-ISI problem and propose a scheme using MDS codes. We show that having identifiable side information is helpful in reducing communication costs for certain cases while maintaining the privacy of desired class identity. Further, we define the problem of PPIR-ISI for multi-user case, where users can collaborately generate the query sets, and propose a scheme for the same.

The system model and problem setting of a PPIR-ISI problem are given in Section $\RN{2}$, where we also give some important notations. In Section $\RN{3}$, we propose a scheme for the PPIR-ISI problem with a single user and demonstrate that identifiable side information can be leveraged to construct an efficient PPIR-SI scheme for a single user. In Section $\RN{4}$, we propose a PPIR-ISI scheme for the multi-user case. The concluding remarks are given in Section $\RN{5}$.

\section{System Model and Problem Setting}
\subsection{System Model}
\begin{figure}
\centering
\begin{tikzpicture}[scale=0.70]
\node[text width=3cm, right=0cm of {2.2,1}] {\footnotesize $F$ messages divided into $\Gamma$ distinct classes.};
\node[text width=3cm, left=0cm of {-0.1,1}] {\footnotesize server};
\node[text width=3cm, left=0cm of {-0.1,-1}] {\footnotesize link};
\node[text width=3cm, left=0cm of {-0.1,-2.5}] {\footnotesize user};
\node[text width=3cm, left=0cm of {-0.1,-3.5}] {\footnotesize side information set, $S$};
\draw[->, thick] (0,0) rectangle (1.5,2); 
\draw[black, thick] (0,0.5) -- (1.5,0.5);
\draw[black, thick] (0,1) -- (1.5,1);
\draw[black, thick] (0,1.5) -- (1.5,1.5);
\draw[->, >=triangle 60][black, thick] (0.75,0) -- (0.75,-2);
\draw[<->][black, thick] (2,0) -- (2,2);
\draw[->, thick] (0.25,-2) rectangle (1.25,-3);
\draw[->, thick] (0,-3) rectangle (1.5,-3.5);
\end{tikzpicture}
\caption{Single-Server PPIR problem containing $F$ messages divided into $\Gamma$ classes with a single user having side information set, $S$.}
\end{figure}
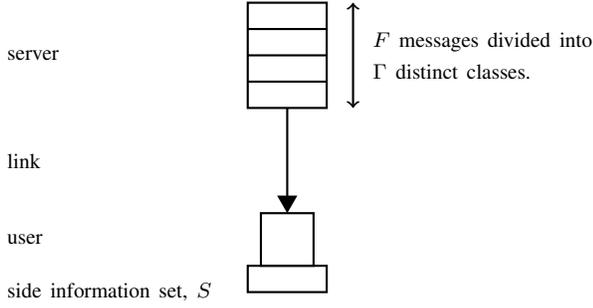
The single-server PPIR-SI problem shown in Fig.1 is described as follows.
We consider a server that stores $F$ messages $W^{(1)}, W^{(2)}, \ldots,W^{(F)}$, where each message $W^{(f)} = (W_1^{(f)}, W_2^{(f)}, \ldots, W_L^{{(f)}}), f \in [F] $ and these $L$ symbols are randomly selected from the field ${F}_q$ with $q$ elements for some $L \in \mathbb{N}$. These $F$ messages are divided into $\Gamma$ distinct classes such that each message belongs exclusively to one class and no class is empty. Without loss of generality, we assume ${\Gamma} \geq 2$.
We define $M_i$ as the set of indices of messages in Class $i$ and $\mu_i$ as the number of messages in Class $i$, i.e., $\mu_i = |M_i|$.
Therefore, the set of messages in Class $i$ is defined as: \{$W^{(f)} | f \in M_i$\}.
Since each message is exclusively present in one class, we have $M_i \cap M_j= \phi$ for $i \neq j$ and $\sum_{i=1}^{\Gamma} \mu_i = {F}$.
For Class $i$, we define a subclass index ${\beta_i}$, such that, ${\beta_i} \in [\mu_i]$. It represents the membership of a message within Class $i$.
The user has a set of side information, denoted by $S$ and the set of side information belonging to Class $i$ as $S_i$. Let $k_i= |S_i|$, then $\sum_{i=1}^{\Gamma} k_i = \kappa$, where $\kappa=|S|$.
\begin{figure}[h]
\centering
\includegraphics[width=0.5\textwidth]{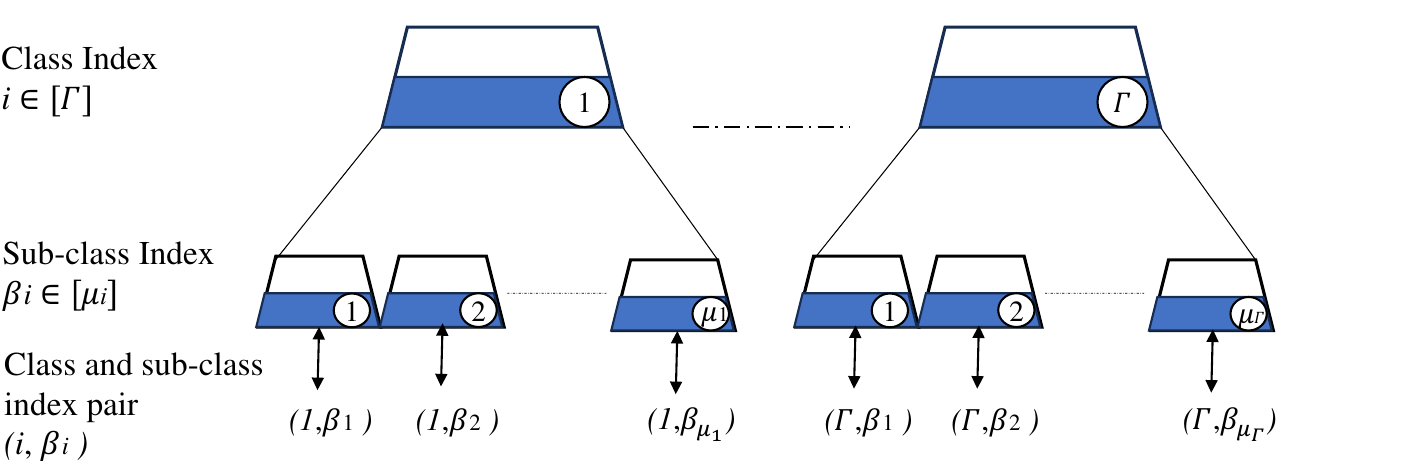}
\caption{Index mapping of $F$ messages classified into $\Gamma$ classes using class-subclass indices.}
\end{figure}
Let $I_{S_i}$ be the subclass index set of the messages belonging to $S_i$. Fig. 2 shows the index mapping of the $F$ messages into $\Gamma$ classes using class-subclass indices, which is illustrated in the following example. In this paper, the notation $Cn$ is used to represent Class $n$, where $n \in [\Gamma]$.

\begin{example}\label{ex1}
 Consider a single server with $F=9$ messages and with $\Gamma =3$ non-overlapping classes,
 \begin{center}
 	C1: \{$W^{(1)},W^{(8)},W^{(7)}\}, $\\
 	C2: \{$W^{(2)},W^{(3)},W^{(5)}\}, $\\
 	C3: \{$W^{(4)},W^{(6)},W^{(9)}\}. $
 \end{center}
Clearly, $\mu_1 = \mu_2 = \mu_3 = 3$, and 
$M_1=\{1,8,7\}$,
$M_2=\{2,3,5\}$,
$M_3=\{4,6,9\}$.
The  message $W^{(5)}$ belongs to Class $2$ and its subclass index, $\beta_i = 3$, thus, class-subclass index pair for $W^{(5)}$ is $(2,3)$. 
Similarly, class-subclass index pairs for all the messages are:
\begin{align*}
W^{(1)} &\rightarrow (1,1), W^{(8)} \rightarrow (1,2), W^{(7)} \rightarrow (1,3),\\
W^{(2)} &\rightarrow (2,1), W^{(3)} \rightarrow (2,2), W^{(5)} \rightarrow (2,3),\\
W^{(4)} &\rightarrow(3,1), W^{(6)} \rightarrow (3,2), W^{(9)} \rightarrow (3,3).
\end{align*}
We can see that, $M_i \cap M_j = \phi \ \forall \ i,j \in [3] \ and \ i \neq j$ and  $\sum_{i=1}^{3} \mu_i = {F}=9$.
Let the side information set be $S =\{W^{(1)},W^{(2)},W^{(6)}\}$. and $S_1 = \{W^{(1)}\}, S_2 = \{W^{(2)}\}, S_3 = \{W^{(6)}\}$. Further the subclass index sets are
$I_{S_1} = \{1\}, I_{S_2} = \{1\}, I_{S_3} = \{2\}$.
\end{example}

\subsection{Problem Setting}
In a single-server PPIR-SI, the user has access to $\kappa$ messages in the side information set, $S$.
In this setup, the user is aware of the size of each class, $\mu_i, \forall i \in [\Gamma]$. The user knows the class index of each message in the side information and is able to divide the side information set $S$ into the subsets $S_i, i \in [\Gamma]$, where $S_i$ is the set of all messages from the side information belonging to Class $i$. Each message in $S_i$ corresponds to a class-subclass index pair, $(i, \alpha)$, where $\alpha \in I_{S_i}$. The identifiability of a class is defined on the basis of subclass indices as follows:

\noindent\textbf{Identifiable class:} If the user knows the subclass index of all its side information belonging to a class, (i.e., the user knows the complete class-subclass index pair), then that class is called an identifiable class.

\noindent\textbf{Unidentifiable class:} If the user is unaware of the subclass index of all its side information belonging to a class, then that class is called an unidentifiable class.

\begin{remark}
There is one more possibility where the user is aware of the subclass index of some of the messages in its side information from a class, but remains unaware of the subclass index of the remaining messages in its side information from that same class. However, in this paper, we are only considering two possibilities: either the user knows all the subclass indices for its side information belonging to a class (called identifiable class), or it knows none of them (called unidentifiable class).
\end{remark}

In this problem, there is $\eta$ number of identifiable classes, where $1 \leq \eta \leq \Gamma$, and the remaining $\Gamma - \eta$ classes are unidentifiable. 
This problem is referred to as a PPIR-ISI problem.
Particularly, if all the classes are identifiable, i.e., $\eta = \Gamma$, then this case is termed as a PPIR-FSI problem and the solution for this case is discussed in Remark \ref{FSI}.

The user privately selects a class index $v$ \ $\in [{\Gamma]}$ and wishes to retrieve any new message from the desired class which is not present in its side information while keeping the requested class identity private.
Given its side information, in order to retrieve a new message
$W^{(n)} \in Cv \backslash S$, i.e., which belongs to the desired class $v$ but not in the side information $S$, the user sends a query $Q^{(v,S)}$ to the server. Query is generated without any prior knowledge of the realizations of the stored messages. In other words,
\[I(W^{(1)}, W^{(2)}, \ldots,W^{{(F)}};Q^{(v,S)})=0.\]
In response to the query, the server responds with an answer $A^{(v, S)}$ to the user. The query and answer should be such that recovery and privacy constraints are satisfied, i.e.,
\begin{equation} \label{RC}
[Recovery]\ H(W^{(n)}|A^{(v,S)},Q^{(v,S)}, v,S)=0.
\end{equation}
\begin{equation}\label{PC}
[Privacy]\ I(v;A^{(v,S)},Q^{(v,S)},W^{[F]})=0.
\end{equation}
Privacy constraint represents that the server cannot learn any information about the desired class from the queries and answers.

\newcommand{\mydefinition}[2]{\textbf{#1}:\space #2}
\noindent\mydefinition{Definition 1}{
The rate for a single server PPIR-SI denoted by $R_{PPIR-SI}$, is defined as the
ratio of the message size $L$ in $F_q$-symbols and the downloaded $F_q$-symbols $D$ to retrieve one message.
\[R_{PPIR-SI} = \frac{L}{D}.\]
}
Given that we are not dividing messages into subparts, the entire message is downloaded rather than just a part of it. Therefore, $D$ is always a multiple of $L$.

\section{MAIN RESULTS} \label{sec:3}
In the PPIR-SI problem, if only the class indices of side information messages are known and the subclass indices of side information messages are unknown, then this case is termed as PPIR with unidentifiable side information (PPIR-USI) \cite{ref6}. The rate achieved by the scheme given for such a case in \cite{ref6} is
\begin{equation}\label{USI}
R_{PPIR-USI}=\left[ {\sum_{i=1}^{\Gamma} min \{k_i +1,\mu_i - k_i\}}\right]^{-1},
\end{equation}
which achieves the capacity of the PPIR-USI problem, i.e., in a PPIR-USI problem, achieving a rate higher than this is not possible.
In the PPIR-ISI problem, we consider applications where the user is partially aware of the side information index set and server structure. We consider that $\eta$ classes out of $\Gamma$ are identifiable. Here, we propose a scheme for PPIR-ISI for the cases with the following assumptions: \\
WLOG the first $\eta$ classes are considered as identifiable and (i) $k_1, k_2, \ldots, k_\eta > k_{un}$, and (ii) $\mu_i -k_{i} \geq \lceil \frac{k_{un}+1}{\eta} \rceil, \ \forall i \in [\Gamma],$ where $k_{un}= \max\{k_{\eta+1},\ldots,k_{\Gamma}\}.$

The rate achieved by the proposed scheme is described in the following theorem.

\begin{theorem}\label{rate}
Consider a single server with F messages divided among $\Gamma$ classes out of which $\eta$ classes are identifiable and Class $i$ contains $\mu_i$ messages, $ i \in [\Gamma]$. The number of messages from Class $i$ present in the side information of user is $k_i, i \in [\Gamma]$, such that $k_1, k_2, \ldots, k_\eta > k_{un}$, and $\mu_i - k_{i} \geq \lceil \frac{k_{un}+1}{\eta} \rceil$, where $k_{un}= \max\{k_{\eta+1},\ldots,k_{\Gamma}\}$, then the following rate can be achieved.
\begin{equation}
R_{PPIR-ISI}= \frac{1}{(k_{un}+1)(\Gamma - \eta +1)}.
\end{equation}
\end{theorem}
The achievability scheme in detail is described in subsection \RN{3} A.

\subsection{Achievability of Theorem 1}
We present a scheme for the single server PPIR-ISI problem which achieves the rate (4) based on MDS codes. The scheme is described as follows.

Given $\eta \leq \Gamma$ is the number of identifiable classes. WLOG first $\eta$ classes are considered to be identifiable, and the remaining are unidentifiable. For a desired class index $v$ $ \in [\Gamma]$ and side information set, $S$, the user sends a series of $k_{un}+1$ queries to the server along with the value of $\eta-1$. The $j^{th}$ query generated by the user is described as $Q^{(j)} (v,S) = \{ (i,\beta^{(j)}_i), \ \forall i \in [\Gamma]\}, \ \forall j \in \{0,1,...,k_{un}\}$, where $\beta^{(j)}_i$ is the subclass index choosing which is given by the Algorithm 1.

Given, $j \in \{0,1, \ldots,k_{un}\}$, the server generates an answer $A^{(j)}(v,S)$: Encoding of the $\ell$th symbol of messages corresponding to the class-subclass indices $ \{ (i,\beta^{(j)}_i), \ \forall i \in [\Gamma]\}$ with a systematic $[2\Gamma - \eta +1 , \Gamma]$ MDS code, $\forall \ \ell \in [L]$ and respond to the user with $(\Gamma - \eta +1)L$ parity symbols. The total number of symbols sent by the server is $(k_{un}+1)(\Gamma - \eta +1) L$.
Thus, we have
\begin{align*}
R_{PPIR-ISI}&=\frac{L}{(k_{un}+1)(\Gamma - \eta +1)L} \\
&= \frac{1}{(k_{un}+1)(\Gamma - \eta +1)}.
\end{align*}

Consider a query $Q^{(j)}(v, S)$ in which messages corresponding to $\eta-1$ class-subclass index pairs are already known to the user. Since, for all $\ell \in [L]$, $(\Gamma - \eta +1)$ parity symbols are sent by the server in answer $A^{(j)}(v,S)$, from the property of MDS code, the user can obtain all $2\Gamma-\eta+1$ symbols. In other words, the user can get all the messages corresponding to the class-subclass index pairs present in a query $Q^{(j)}(v, S)$ if $\eta-1$ class-subclass index pairs are already known to the user.
%

\begin{algorithm}
    \SetKwProg{Function}{Function}{}{}
    
    \Function{id($i,j,v$)}{
	Select a random variable $r \in \{0,1,\ldots, k_{un}\}$. \\
        \For{i=1 \KwTo $\eta$}{
           \uIf{$j = r $}{
                $\beta^{(r)}_v \in [\mu_v]\backslash I_{S_v}$ \\
                $\beta_{i}^{(r)} \in I_{S_{i}} \text{  for $i \neq v$ }$ $ \ \ \ \cdots \cdots\cdots \cdots$(i)
            }
            \Else{
                $\beta^{(j)}_i \in [\mu_i]\backslash \  \bigcup\limits_{m=0}^{j-1} \{\beta^{(m)}_i\} $ $ \ \ \ \cdots \cdots$(ii)
            }
        }
        \For{$i=\eta +1$ \KwTo $\Gamma$}{
            $\beta^{(j)}_i \in [\mu_i]\backslash \  \bigcup\limits_{m=0}^{j-1} \{\beta^{(m)}_i \} $ $ \ \ \ \cdots \cdots\cdots $(iii)
        }
    }

    \SetKwProg{Function}{Function}{}{}
    \Function{unid($i,j$)}{
    Define $t=j\bmod{\eta} + 1$ \\
        \For{i=1 \KwTo $\eta$}{
           $\beta^{(j)}_t \in [\mu_t]\backslash ( I_{S_t}\bigcup\limits_{m=0}^{j-1}\{\beta^{(m)}_t\})$   $ \ \ \ \cdots \cdots \cdots$(iv)\\
           $\beta_{i}^{(j)} \in I_{S_{i}} \backslash { \bigcup\limits_{m=0}^{j-1} \{\beta_{i}^{(m)}\}, } \text{  for $i \neq t$}$ $ \ \ \ \cdots\cdots $(v)
            }
        
        \For{$i=\eta +1$ \KwTo $\Gamma$}{
            $\beta^{(j)}_i \in [\mu_i]\backslash \  \bigcup\limits_{m=0}^{j-1} \{\beta^{(m)}_i\} $ $ \ \ \cdots \cdots\cdots \cdots$(vi)
        }   
     }

    \SetKwProg{Function}{Function}{}{}
    
    \Function{Main}{
        
        \For{$j= 0$ \KwTo $k_{un}$}{

            \uIf{$ v \ \in [\eta]$}{
                Call \Function{id($i,j,v$)} \
            }
            \Else{
            Call \Function{unid($i,j$)} \
            }            
        }
        
        \Return{Final result}\;
    }
    Call \Function{Main}\

    \caption{$\beta^{(j)}_i$ for Single Server PPIR-ISI Scheme}
\end{algorithm}

\begin{proof}[\textbf{Proof of correctness}]
In Algorithm 1, Function $id$ generates the subclass indices in queries if the desired class is identifiable, and Function $unid$ generates the subclass indices in queries if the desired class is unidentifiable. If the desired class is identifiable, then only one query out of $k_{un}+1$ queries is sufficient to get a new message from the desired class. In Function $id$, the number $r$ is randomly selected from $\{0,1, \ldots, k_{un}\}$. In $r$-th query $Q^{(r)}(v, S)$, an unknown subclass index $\beta_v^{(r)}$ is sent from the desired class $v$ and known subclass indices $\beta_i^{(r)}$ are sent from other $\eta-1$ identifiable classes $i \in [\eta], i \neq v$ (refer to (i)). Since in query $Q^{(r)}(v, S)$, the total $\eta-1$ class-subclass index pairs are already known to the user, all the messages corresponding to that query can be obtained. Therefore, the user will get a new message from Class $v$.
The remaining $k_{un}$ queries are sent to maintain the uncertainty between identifiable and unidentifiable classes, and in all these queries $Q^{(j)}(v, S), j \neq r$, a new subclass index is sent from each class.

If the desired class $v$ is unidentifiable, then at least $k_v+1$ messages are required from Class $v$ to guarantee a new message. The user generates a series of $k_{un} + 1$ queries. In each query, $\eta-1$ known class-subclass index pairs, i.e., class-subclass index pairs corresponding to the messages available in the side information of the user, are sent (refer to (v)). Therefore, all the messages corresponding to each query can be obtained, and the user will get $k_{un}+1$ messages from each unidentifiable class.
\end{proof}

\begin{remark}The value of $\eta -1$ is sent to the server along with the queries to set up the dimension of the MDS code.
\end{remark}

\begin{proof}[\textbf{Proof of privacy}] From Algorithm 1, it is clear that for each class $i \in [\Gamma]$, no class-subclass index pair is repeated in all $k_{un}+1$ queries, for any realization of $v \in [\Gamma]$. Therefore, the server does not get any information about the desired class from the query $Q^{(j)}(v,S)$ and the answer $A^{(j)}(v,S), i \in [\Gamma]$. Hence, the PPIR-ISI scheme satisfies the privacy constraint \eqref{PC}.
\end{proof}

The justifications for the assumptions in Theorem 1 are:
\begin{enumerate}
\item Queries are generated over $k_{un}+1$ times, where $k_{un}$ is defined as the $\max\{k_{\eta+1},\ldots,k_{\Gamma}\}$. If the desired class is identifiable, a new message from the desired class is recovered in the first query. If the desired class is unidentifiable, then according to Algorithm 1, identifiable classes require at most $k_{un}+1$ messages as its side information (refer to eq. (v) in Algorithm 1). This means that the side information of identifiable classes should be at least $k_{un}+1$. Hence the strict inequality, $k_1,k_2,\ldots,k_\eta > k_{un}$.
\item In the proposed scheme, if the desired class is unidentifiable, then from (iv) in Algorithm 1, one unknown class-subclass index is sent from one out of $\eta$ identifiable classes in each query. That means there is at least one identifiable class from which the exact number of unknown class-subclass index present in the query set is $\left \lceil \frac{k_{un}+1}{\eta} \right \rceil$. Hence, there should be at least $\left \lceil \frac{k_{un}+1}{\eta} \right \rceil$ number of messages in each class which are not present in the side information of the user, i.e.,
$\mu_i - k_{i} \geq \left \lceil \frac{k_{un}+1}{\eta} \right \rceil$.\\

\end{enumerate}

\begin{remark}
The idea of the proposed scheme is motivated by the conjecture given in \cite[Remark 2]{ref6}. However, it is hard to achieve the rate $R=\frac{1}{\Gamma-\eta+1}$ given in that conjecture for a PPIR-ISI problem. If the desired class $v$ is unidentifiable, then at least $k_v+1$ messages are required to get a new message from the desired class, which is not possible in one query $Q(v,S) = \{ (i,\beta_i), \ \forall i \in [\Gamma]\}$.
\end{remark}

\begin{remark}\label{FSI} Consider the problem of PPIR-FSI, in which all the classes are identifiable, i.e., $\eta = \Gamma$. Since there are no unidentifiable classes, we have $k_{un}=0$. Therefore, from Theorem 1, the rate achieved by the proposed scheme for the PPIR-FSI case is $R_{PPIR-FSI} =1.$ The same idea for PPIR-FSI problem was given in \cite[Remark 2]{ref6}.
\end{remark}

\begin{example} \label{ex2} Consider a single server with $\Gamma=5$ non-overlapping classes out of which $\eta=3$ classes are identifiable to the user. The total $39$ messages are distributed  in $\Gamma = 5$ classes, each message containing L = 2 symbols from $F_{11}.$ 

\begin{align*}
C1: &W^{(4)},W^{(7)},W^{(11)},W^{(12)},W^{(21)},W^{(28)},W^{(30)},\\
C2: &W^{(1)},W^{(8)},W^{(13)},W^{(17)},W^{(25)},W^{(39)},\\
C3: &W^{(2)},W^{(10)},W^{(15)},W^{(19)},W^{(23)},W^{(27)},W^{(31)},W^{(34)},\\
C4: &W^{(3)},W^{(9)},W^{(16)},W^{(18)},W^{(24)},W^{(26)},W^{(32)},W^{(35)},\\
 & W^{(37)},\\
C5:  &W^{(5)},W^{(6)},W^{(14)},W^{(20)},W^{(22)},W^{(29)},W^{(33)},W^{(36)},\\
& W^{(38)}.
\end{align*}
The class-subclass index set, i.e., $\{ (i,\beta_i),\ \forall \ \beta_i\in [\mu_i]\}$ corresponding to each message in C1 is 
$\{(1,1),(1,2),(1,3),(1,4),$ $(1,5),(1,6),(1,7)\},$ respectively.
Similarly, a class-subclass index is assigned to each message in C2, C3, C4, C5. Also,
$\mu_1=7,\mu_2=6,\mu_3=8,\mu_4=9,\mu_5=9$. Let the side information set be $S=\{W^{(11)},W^{(12)},W^{(30)},W^{(1)}, W^{(8)}, W^{(17)},\\ W^{(25)},W^{(2)},W^{(10)},W^{(15)},W^{(19)},W^{(27)},W^{(9)},W^{(16)}, W^{(5)},\\  W^{(6)}, W^{(36)}\}$.
Observe that, $k_1= 3, k_2=4, k_3=5, k_4=2, k_5=3$. WLOG, we take first three classes as identifiable and the remaining $4^{th} \text{ and } 5^{th}$ as unidentifiable, i.e., the user knows the sets $I_{S_1}=\{3,4,7\}, I_{S_2}=\{1,2,4,5\}$, $I_{S_3}=\{1,2,3,4,6\}$, and the sets $I_{S_4}=\{2,3\}$ and $I_{S_5}=\{1,2,8\}$ are not known to the user. Clearly, $k_{un} = \max\{k_4,k_5\} = 3.$
The user sends the value of $\eta-1=2$ and the following queries to the server.
$$Q^{(j)} (v,S) = \{ (i,\beta^{(j)}_i), \forall \ i \in [5] \}, \  \forall \ j \in \{0,1,2,3\}.$$ 
\noindent\textit{\textbf{Case 1}}: Let $v=3$ be the desired class and the random selected number $r=0$.
Using Algorithm 1, following queries are generated.
\begin{align*}
Q^{(0)} (v,S) &= \{ (1,3),(2,2),(3,5),(4,8),(5,3)\},\\
Q^{(1)} (v,S) &= \{ (1,6),(2,4),(3,7),(4,2),(5,5)\},\\
Q^{(2)} (v,S) &= \{ (1,4),(2,1),(3,6),(4,9),(5,1)\},\\
Q^{(3)} (v,S) &= \{ (1,5),(2,3),(3,2),(4,1),(5,2)\}.
\end{align*}
From the above query set, we can see that no index of a particular class is repeated, thus making the server oblivious to the demand class of the user. For all $ j \in \{0,1,2,3\}$, the server responds to the user with the answer $A^{(j)}(v,S)$: 
Encoding of each symbol of the messages corresponding to the class-subclass index pairs $\{ (i,\beta^{(j)}_i) \ \forall \ i \in [5] \}$ with a systematic $[8,5]$ MDS code and respond to the user with $3$ parity symbols. Therefore, the rate is $R=\frac{1}{4 \times 3}=\frac{1}{12}.$ 

To compute $A^{(0)}(v,S)$, consider the messages corresponding to the class-subclass index pairs in query $Q^{(0)}(v, S)$ given as follows.\\
$(1,3) \rightarrow W^{(11)}=\{0,  1\},$
$(2,2) \rightarrow W^{(8)}=\{1,  7\},$\\
$(3,5) \rightarrow W^{(23)}=\{ 9, 4\},$
$(4,8) \rightarrow W^{(35)}=\{6, 1\},$\\
$(5,3) \rightarrow W^{(14)}=\{8, 3\}.$\\
Let $m_1$ denote the vector containing all the first symbols of the messages corresponding to the query $Q^{(0)} (v,S)$, i.e.,
$m_1= \begin{bmatrix}

            0 &  1 & 9 & 6 & 8

        \end{bmatrix},$
and $m_2$ denote the vector containing all the second symbols of the messages corresponding to the query $Q^{(0)} (v,S)$, i.e.,
$m_2= \begin{bmatrix}

            1 &  7 & 4 & 1 & 3

        \end{bmatrix}.$
Suppose the following generator matrix of a $[8, 5]$ MDS code over $F_{11}$ is used by the server for encoding.
$$G_{5 \times 8}= \begin{bmatrix}

            1 &  0 & 0 & 0 & 0 & 1 & 5 & 4 \\
            0 &  1 & 0 & 0 & 0 & 6 & 9 & 7\\
            0 &  0 & 1 & 0 & 0 & 10 & 1 & 5 \\
            0 &  0 & 0 & 1 & 0 & 1 & 4 & 5 \\
            0 &  0 & 0 & 0 & 1 & 5 & 4 & 2

        \end{bmatrix}.$$
   Let $c_i = m_i  G$ denote the codeword generated using $[8, 5]$ MDS code for the $i^{th}$ symbols of the messages corresponding to the query $Q^{(0)} (v,S)$. Then we have,
    $$c_1= \begin{bmatrix}

            0 &  1 & 9 & 6 & 8 & 10 & 8 & 10

        \end{bmatrix},$$
and
        $$c_2= \begin{bmatrix}

            1 &  7 & 4 & 1 & 3 & 0 & 0 & 7

        \end{bmatrix}.$$
       The server will send 3 parity symbols $\{10, 8, 10\}$ and $\{ 0, 0, 7 \}$ from each $c_1$ and $c_2$, respectively. Thus, in answer $A^{(0)} (v,S)$,  total $3 \times 2=6$ parity symbols are sent to the user. Similarly, the server will compute $A^{(j)}(v,S) \ \text{for remaining} \ j \ \in \ \{1,2,3\}$ and sent it to the user.

Since the user already has $W^{(11)}$ and $W^{(8)}$ in its side information, the $1$st and $2$nd symbols of $c_1$ are already known to the user. Now the user knows $1, 2, 6, 7$ and $8$th symbols of $c_1$, and can get the whole vector $m_1$ as follows.
$$m_1 G'=m_1  \begin{bmatrix}

            1 &  0  & 1 & 5 & 4 \\
            0 &  1  & 6 & 9 & 7\\
            0 &  0  & 10 & 1 & 5 \\
            0 &  0  & 1 & 4 & 5 \\
            0 &  0  & 5 & 4 & 2

        \end{bmatrix}= \begin{bmatrix}

            0 &  1 & 10 & 8 & 10

        \end{bmatrix},$$
where $G'$ is a submatrix of $G$ consituiting $1,2,6,7$ and $8$th columns of $G$, and it is an invertible matrix from the property of an MDS code. Therefore, we have
\begin{align*}
m_1&=  \begin{bmatrix}
            0 &  1 & 10 & 8 & 10
        \end{bmatrix}\begin{bmatrix}
            1 &  0  & 6 & 7 & 4 \\
            0 &  1  & 1 & 5 & 9\\
            0 &  0  & 1 & 4 & 4 \\
            0 &  0  & 10 & 5 & 1 \\
            0 &  0  & 5 & 2 & 5
        \end{bmatrix}\\
&= \begin{bmatrix}
            0 &  1 & 9 & 6 & 8
        \end{bmatrix}.
\end{align*}
Similarly, the user will compute $m_2$ and get the message $W^{(23)}=\{9,4\}$ from the desired class C3.

\noindent\textit{\textbf{Case 2}}: Let $v=4$ be the desired class. Using Algorithm 1, following queries are generated.
\begin{align*}
Q^{(0)} (v,S) &= \{ (1,1),(2,1),(3,2),(4,9),(5,1)\},\\
Q^{(1)} (v,S) &= \{ (1,4),(2,6),(3,1),(4,8),(5,3)\},\\
Q^{(2)} (v,S) &= \{ (1,7),(2,5),(3,7),(4,2),(5,5)\},\\
Q^{(3)} (v,S) &= \{ (1,2),(2,4),(3,3),(4,1),(5,2)\}.
\end{align*}
From the above query set, we can see that no index of a particular class is repeated, thus making the server oblivious to the demand class of the user.
 For all $ j \in \{0,1,2,3\}$, the server responds to the user with the answer $A^{(j)}(v,S)$: 
Encoding of each symbol of the messages corresponding to the class-subclass index pairs $\{ (i,\beta^{(j)}_i) \ \forall \ i \in [5] \}$ with a systematic $[8,5]$ MDS code over $F_{11}$ and respond to the user with $3$ parity symbols.

In query $Q^{(0)} (v,S)$, the user already knows the messages $W^{(1)}$ and $W^{(10)}$ corresponding to the class-subclass index pairs $(2,1)$ and $(3,2)$, respectively. After receiving $3$ parity symbols, for each $\ell=1,2$, in the answer $A^{(0)}(v,S)$, the user will obtain all the messages corresponding to the $Q^{(0)} (v,S)$, i.e., $W^{(4)}, W^{(1)}, W^{(10)}, W^{(37)}$ and $W^{(5)}$.
Similarly, the user will obtain all the messages corresponding to other queries $Q^{(j)} (v,S), j =1,2,3$, and the messages obtained from Class $4$ are $W^{(37)}, W^{(35)}, W^{(9)}$ and $W^{(3)}$.
\end{example}

\textbf{Note:} In Example \ref{ex2}, if we consider all classes as unidentifiable, and use the scheme given in \cite{ref6}, then the achieved rate is $R_{PPIR-USI}=\frac{1}{16} < \frac{1}{12}.$ This shows the improvement in the rate due to identifiability. 

\begin{remark}

While the proposed scheme provides privacy for the identity of the desired class, the privacy of side information is not fully achieved. For instance, consider Example \ref{ex2}. In this scenario, the server lacks precise knowledge of which messages are included in the user's side information. However, it is aware of the following:
(1) If the desired class is identifiable, then out of the 20 class-subclass indices present in the queries, at least 2 correspond to messages from the side information, 
(2) If the desired class is unidentifiable, then out of the 20 class-subclass indices present in the queries, at least 8 correspond to messages from the side information.
Given that the server remains unaware of any details regarding the identity of the desired class, either of the two cases mentioned above could be true.
\end{remark}

In the PPIR-USI scheme \cite{ref6}, encoding of the $\ell^{th}$ symbol is done for all the messages within Class $i$, whereas, in the proposed PPIR-ISI Scheme, encoding of the $\ell^{th}$ symbol is done for the messages across all the classes. Due to this, it is hard to give a direct comparison between the rate of these two schemes due to the summation present in \eqref{USI}

For some cases, it can be shown analytically that having identifiable side information improves the rate, i.e.,  $R_{PPIR-ISI} \geq R_{PPIR-USI}$. These cases are presented in the  following three theorems.

\begin{theorem}\label{compare}
	Consider a single server with F messages divided among $\Gamma$ classes out of which $\eta$ classes are identifiable, and each Class $i$ contains $\mu_i$ messages, $ i \in [\Gamma]$. The side information from Class $i$ present with the user is $k_i, i \in [\Gamma]$ such that $k_1, k_2, \ldots, k_\eta > k_{un}$, and $\mu_i - k_{i} \geq \lceil \frac{k_{un}+1}{\eta} \rceil$, where $k_{un}= \max\{k_{\eta+1},\ldots,k_{\Gamma}\}$, then the proposed scheme for PPIR-ISI performs better than the PPIR-USI scheme \cite{ref6}, if
	$
	k_i +1 \leq \mu_i -k_i, \ \forall \ i \ \in \Gamma \ \text{ and} \
	k_{un} = k_i,   \ \forall \ i \ \in \{\eta+1,\ldots,\Gamma\} .
	$
\end{theorem}

\begin{proof}
The proposed scheme works better than the PPIR-USI scheme \cite{ref6} when
$\frac{R_{PPIR-ISI}}{R_{PPIR-USI}} \geq 1 .$
Considering $
k_i +1 \leq \mu_i -k_i, \ \forall \ i \ \in \Gamma \ \text{ and} \
k_{un} = k_i,   \ \forall \ i \ \in \{\eta+1,\ldots,\Gamma\}
$, 
we have $R_{PPIR-USI}= \frac{1}{\sum_{i=1}^{\Gamma}(k_i+1)}$ and
\begin{align*}
\frac{R_{PPIR-ISI}}{R_{PPIR-USI}} &= \frac{\sum_{i=1}^{\Gamma}(k_i+1)}{(k_{un}+1)(\Gamma-\eta+1)} \\
&= \frac{1}{\Gamma-\eta+1} \times \left[\frac{k_1+1}{k_{un}+1}+\cdots+\frac{k_\eta+1}{k_{un}+1} \right. \\
&\quad \left.+\frac{k_{\eta+1}+1}{k_{un}+1} + \cdots+\frac{k_\Gamma+1}{k_{un}+1}\right]\\ 
& \geq \frac{\Gamma}{\Gamma -\eta +1} \geq 1. \qquad \text{(as $k_i \geq k_{un}, \forall i \in [\Gamma]$)} 
\end{align*}
\end{proof}

Now we present the other two cases, for which, we show that having identifiable side information is advantageous, i.e., $R_{PPIR-ISI} \geq R_{PPIR-USI}$.  For these two cases the conditions taken in Theorem \ref{rate} are also assumed in the the following two theorems.
\begin{itemize}
\item $k_1, k_2, \ldots, k_\eta > k_{un}$, 
\item $\mu_i -k_{i} \geq \left\lceil \frac{k_{un}+1}{\eta} \right\rceil, \ \forall i \in [\Gamma].$ 
\end{itemize}

\begin{theorem}\label{compare2}
If $k=k_i, \ \forall i \in [\eta]$,  $k_{un}=k_i, \ \forall i \in \{\eta+1, \ldots, \Gamma \}$, and $\mu_i=\mu$, for all $i \in [\Gamma]$ such that $k_i+1 \geq \mu -k_i$, then the proposed scheme works better than the PPIR-USI scheme\cite{ref6}, when $\mu \geq k + \left \lceil \frac{(\Gamma - \eta +1)(k_{un}+1)}{\Gamma } \right \rceil$. 
\end{theorem}

\begin{proof} For the given conditions, we have
\begin{align*}
\frac{R_{PPIR-ISI}}{R_{PPIR-USI}}&= \frac{{\sum_{i=1}^{\Gamma} (\mu - k_i)}}{(k_{un}+1)(\Gamma -\eta +1)}\\
&=\frac{\eta (\mu - k)+(\Gamma - \eta) (\mu - k_{un})}{(k_{un}+1)(\Gamma -\eta +1)}\\
&\geq \frac{\Gamma (\mu - k)}{(k_{un}+1)(\Gamma -\eta +1)} \qquad \text{(as $k >k_{un}$)} \\
& \geq 1.  \qquad \left(\text{as $\mu -k \geq  \frac{(\Gamma - \eta +1)(k_{un}+1)}{\Gamma } $}\right)
\end{align*}
Hence, we have $R_{PPIR-ISI} \geq R_{PPIR-USI}$.
\end{proof}

\begin{theorem}\label{compare3}
If $k_i+1 \geq \mu_i -k_i$, for all $i \ \in [\Gamma]$, and only one class is identifiable, i.e., $\eta=1$ then the proposed scheme works better than the PPIR-USI scheme\cite{ref6}.
\end{theorem}

\begin{proof} We know that for all $i \in [\Gamma]$,   
$$\mu_i -k_{i} \geq \left\lceil \frac{k_{un}+1}{\eta} \right\rceil =k_{un}+1.$$
 We have
\begin{align*}
\frac{R_{PPIR-ISI}}{R_{PPIR-USI}}&= \frac{ {\sum_{i=1}^{\Gamma} (\mu_i - k_i)}}{(k_{un}+1)(\Gamma -\eta +1)} \\
&=\frac{1}{\Gamma -\eta +1} \sum_{i=1}^{\Gamma} \frac{\mu_i-k_i}{k_{un}+1}\\
& \geq \frac{1}{\Gamma -\eta +1} \sum_{i=1}^{\Gamma} 1 \quad \text{(as $\mu_i -k_{i} \geq k_{un}+1$)}\\
&=\frac{\Gamma}{\Gamma -\eta +1}=1.
\end{align*}
\end{proof}

It is to be noted that Example \ref{ex2} does not fall into any of the cases given in Theorems \ref{compare}, \ref{compare2} and \ref{compare3}. Clearly, the rate of the proposed PPIR-ISI scheme is greater than the rate of the PPRIR-USI scheme \cite{ref6} in Example \ref{ex2}. This indicates that apart from the three cases specified in Theorems \ref{compare}, \ref{compare2} and \ref{compare3}, there exist other scenarios in which the proposed PPIR-ISI scheme performs better than the PPIR-USI scheme.

In the following example, the parameters of PPIR-ISI satisfy the conditions given in Theorem \ref{compare}, and  the proposed scheme performs better than the PPIR-USI scheme given in \cite{ref6}.

\begin{example} \label{ex3}
Consider a single server with $\Gamma=6$ non-overlapping classes out of which $\eta=3$ classes are identifiable to the user. The total 49 messages are distributed  in $\Gamma = 6$ classes, and each message is containing L = 2 symbols from $F_{11}.$
\begin{align*}
C1:  &W^{(1)},W^{(6)},W^{(11)},W^{(21)},W^{(26)}, W^{(35)}, W^{(41)}, W^{(43)},\\ & W^{(49)},\\
C2: &W^{(2)},W^{(7)},W^{(12)},W^{(17)},W^{(22)},W^{(36)}, W^{(38)},W^{(44)}, \\ & W^{(47)},\\
C3: &W^{(3)},W^{(13)},W^{(18)},W^{(28)},W^{(33)},W^{(37)}, W^{(39)},W^{(40)}, \\ & W^{(42)}, W^{(46)}\\
C4: &W^{(9)},W^{(14)},W^{(19)},W^{(24)},W^{(29)}, W^{(48)},\\
C5:  &W^{(5)},W^{(10)},W^{(15)},W^{(20)},W^{(25)},W^{(30)},W^{(34)},\\
C6: &W^{(4)},W^{(8)}, W^{(16)}, W^{(23)}, W^{(27)}, W^{(31)}, W^{(32)}, W^{(45)}.
\end{align*}
The class-subclass index set corresponding to each message in C1 is 
$\{(1,1),(1,2),(1,3),(1,4),(1,5),(1,6),(1,7), (1,8), (1,9)\},$ respectively.
Similarly, a class-subclass index is assigned to each message in C2, C3, C4, C5, C6. Also,
$\mu_1=9,\mu_2=9,\mu_3=10,\mu_4=6,\mu_5=7, \mu_6=8$. Let the side information set be  $S=\{W^{(6)},W^{(11)},W^{(21)},W^{(41)}, W^{(2)},\\ W^{(12)},  W^{(17)}, W^{(44)},W^{(3)},W^{(33)},W^{(39)},W^{(14)},W^{(19)}, \\ W^{(10)},
W^{(25)},  W^{(23)},W^{(45)}\}$.
Observe that, $k_1= 4, k_2=4, k_3=3, k_4=2, k_5=2, k_6=2$. WLOG, take the first three classes as identifiable and the remaining $4^{th}, 5^{th} \text{ and } 6^{th}$ as unidentifiable, i.e., the user knows the sets $I_{S_1}=\{2,3,4,7\}$, $I_{S_2}=\{1,3,4,8\}$, $I_{S_3}=\{1,5,7\}$, and the sets  $I_{S_4}=\{2,3\}$, $I_{S_5}=\{2,5\}$ and $I_{S_6}=\{4,8\}$ are not known to the user. Clearly, $ k_4=k_5=k_6=2,$ and thus, $k_{un} = 2.$
Let $v=4$ be the desired class. Then the user sends the value of $\eta-1=2$ and following queries to the server.  
\begin{align*}
Q^{(0)} (v,S) &= \{ (1,9),(2,3),(3,5),(4,3),(5,2),(6,8)\},\\
Q^{(1)} (v,S) &= \{ (1,4),(2,2),(3,7),(4,6),(5,1),(6,1)\},\\
Q^{(2)} (v,S) &= \{ (1,7),(2,8),(3,10),(4,2),(5,3),(6,6)\}.
\end{align*}
From the above query set, we can see that no subclass index from a particular class is repeated, thus making the server oblivious to the demand class of the user.
 For all $ j \in \{0,1,2\}$, the server responds to the user with the answer $A^{(j)}(v,S)$: 
Encoding of each symbol of the messages corresponding to the class-subclass index pairs $\{ (i,\beta^{(j)}_i) \ \forall \ i \in [6] \}$ with a systematic $[10,6]$ MDS code over $F_{11}$ and respond to the user with $4$ parity symbols. Therefore, the rate is 
$$R_{PPIR-ISI}=\frac{1}{3 \times 4}=\frac{1}{12}.$$
In query $Q^{(0)} (v,S)$, the user already knows the messages $W^{(12)}$ and $W^{(33)}$ corresponding to the class-subclass index pairs $(2,3)$ and $(3,5)$, respectively. After receiving $4$ parity symbols, for each $\ell=1,2$, in the answer $A^{(0)}(v,S)$, the user will obtain all the messages corresponding to the $Q^{(0)} (v,S)$, i.e., $W^{(49)}, W^{(12)}, W^{(33)}, W^{(19)}, W^{(10)}$ and $W^{(45)}$.
Similarly, the user will obtain all the messages corresponding to other queries $Q^{(j)} (v,S), j =1,2$, and the messages obtained from Class 4 are $W^{(19)}, W^{(48)}$ and $W^{(14)}$, in which $W^{(48)}$ is not present in the side information of the user.\\

\end{example}

\textbf{Note:} In Example \ref{ex3}, if we consider all classes as unidentifiable, and use the scheme given in \cite{ref6}, then the achieved rate is $$R_{PPIR-USI}=\frac{1}{\sum_{i=1}^{6} (k_i+1)}=\frac{1}{23} < \frac{1}{12}.$$

\section{Single-Server PPIR-ISI for Multi-User Scenario}
In the previous section, we considered the PPIR-ISI problem, where there was only one user who was partially aware of the side information and the server structure. 
In this section, we extend the problem to a multi-user setting and propose a scheme, and establish the rate achieved by it.\\

\noindent\mydefinition{Definition 2}{
The rate for single server PPIR-SI with $U$ users is defined as the ratio of the message size $L$ in $F_q$-symbols and the downloaded $F_q$-symbols $D$ to retrieve one desired message by each user, i.e., 
\[R_{PPIR-SI} = \frac{L}{D}.\]

}

Consider a server with $\Gamma$ classes out of which $\eta$ classes are identifiable to each user $u \ \in \ [U]$. Each user has a desired class, denoted by $v_u, \ \forall \ u \ \in \ [U]$.
One possible solution to this problem is that the demands of all the users are satisfied by sending a separate query set for each user using the PPIR-ISI scheme given in Section \RN{3}. 
 Since there are $U$ users, Algorithm 1 will run $U$ times. The rate for this case will be 
\begin{equation}
R_{PPIR-ISI}= \frac{1}{U(k_{un}+1)(\Gamma - \eta +1)},
\end{equation}
assuming that the parameter $k_{un}$ is same of each user $u \in [U]$. 
The total number of downloaded messages will be very high. Hence, we develop the Multi-user PPIR-ISI scheme for the system model described in Fig. 3, in which users can collude among themselves and generate a query set cooperatively.\\

\begin{figure}[h]
\centering
\begin{tikzpicture}[scale=0.70]
    \draw[->, thick] (0,0) rectangle (1.5,2);
    \draw[black, thick] (0,0.5) -- (1.5,0.5);
    \draw[black, thick] (0,1) -- (1.5,1);
    \draw[black, thick] (0,1.5) -- (1.5,1.5);
    \draw[->, >=triangle 60][black, thick] (0.75,0) -- (0.75,-3);
    \draw[<->][black, thick] (2,0) -- (2,2);
    \node[text width=3cm, right=0cm of {2.2,1}] {\footnotesize F messages divided into $\Gamma$ distinct classes.};

    \draw[->, thick] (0.25,-3) rectangle (1.25,-4);
    \draw[->, thick] (0,-4) rectangle (1.5,-4.5);

    \draw[->, thick] (3,-3) rectangle (4,-4);
    \draw[->, thick] (2.75,-4) rectangle (4.25,-4.5);

    \draw[->, thick] (-2.5,-3) rectangle (-1.5,-4);
    \draw[->, thick] (-2.75,-4) rectangle (-1.25,-4.5);
    
    \node[text width=3cm, left=0cm of {-1.5,1}] {\footnotesize server};
    \node[text width=3cm, left=0cm of {-1.5,-1}] {\footnotesize shared link};
    \node[text width=3cm, left=0cm of {-1.5,-3.25}] {\footnotesize U users};
    \node[text width=2cm, left=0cm of {-3,-4.5}] {\footnotesize side information set, ${S^u}$};
    
    \draw[->, >=triangle 60] (0.75, -1.05) -- (-2, -3);
    \draw[->, >=triangle 60] (0.75, -1.05) -- (3.5, -3);
    
\end{tikzpicture}
\caption{The Single-Server PPIR problem containing F files divided into $\Gamma$ classes with U, each having side information set $S^u, \ \forall \ u \in [U]$.}
\end{figure}
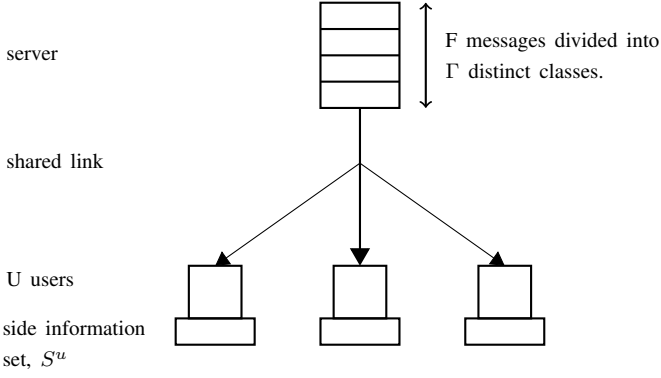

\subsection{Multi-user PPIR-ISI scheme}
Consider a single server with $\Gamma$ non-overlapping classes. There are $U$ users, and each user has a side information set denoted by $S^{u}, \  \forall  \ u  \in [U]$, where $S^u = \bigcup_{i=1}^{\Gamma} S_i^u$ and $S_i^u$ denotes the set of messages from class $i$ in the side information set of user $u$. Let $I_{S_i^u}$ denote the subclass indices of files from class $i$ in the side information set of user $u$.
To illustrate the Multi-user scheme, we consider the following example:
\begin{example}\label{ex4}
Consider a single server with $\Gamma =3$ non-overlapping classes,
\begin{center}
C1: \{$W^{(1)},W^{(8)},W^{(7)}\}, $\\
C2: \{$W^{(2)},W^{(3)},W^{(5)}\}, $\\
C3: \{$W^{(4)},W^{(6)},W^{(9)}\}. $\\
\end{center}
Consider the message $W^{(3)}$ message.  It belongs to Class 2 and its subclass index, $\beta_i = 2$, thus, the class-subclass index pair for $W^{(3)}$ is $(2,2)$. 
Similarly, class-subclass index pairs for all the messages are:
\begin{align*}
W^{(1)} &\rightarrow (1,1), W^{(8)} \rightarrow (1,2), W^{(7)} \rightarrow (1,3),\\
W^{(2)} &\rightarrow (2,1), W^{(3)} \rightarrow (2,2), W^{(5)} \rightarrow (2,3),\\
W^{(4)} &\rightarrow(3,1), W^{(6)} \rightarrow (3,2), W^{(9)} \rightarrow (3,3).
\end{align*}
Given, the number of users are $U = 2$, then side information set $S^u$ for user $u \ \in [2]$ is,
$S^1=\{W^{(1)},W^{(2)},W^{(6)}\}$ and $S^2=\{W^{(8)},W^{(5)},W^{(9)}\}$. Clearly, for User $1$, $S_1^1 =\{W^1\} , S_2^1 =\{W^2\}, S_3^1 =\{W^6\} \text{ and for User 2},\ S_1^2 = \{W^8\}, S_2^2 = \{W^5\}, S_3^2 = \{W^9\} $. Also, the subclass indices of messages for User 1 are, $I_{S_1^1} =\{1\}, I_{S_2^1} =\{1\}, I_{S_3^1} =\{2\} \text{ and for User }2 \text{ are }, I_{S_1^2} = \{2\}, I_{S_2^2} = \{3\}, I_{S_3^2} = \{3\} $.
\end{example}

WLOG the first $\eta$ classes are considered as identifiable. The following assumptions are made: With $k_{un}^u= \max\{k_{\eta+1}^u,\ldots,k_{\Gamma}^u\}$ and $k_{un}=\max\{k_{un}^1,\ldots,k_{un}^U\}$ we have  
\begin{enumerate}
\item $k_1^u, k_2^u, \ldots, k_\eta^u >k_{un} \ \forall \ u \ \in \ [U]$ , 
\item $k_{un} + 1 \geq U$, 
\item For all $i \ \in \ [\eta]$ and $u \ \in \ [U], \ \mu_i - \ k_i^u \geq k_{un}+1$.
\end{enumerate}
The rate achieved by the proposed scheme for multi-users is described in the following theorem.
\begin{theorem} 
	\label{thm5}
    Consider a single server with F messages divided among $\Gamma$ classes out of which $\eta$ classes are identifiable and Class $i$ contains $\mu_i$ messages, $ i \in [\Gamma]$. The total number of side information messages from Class $i$ present with the user $u$ is $k_i^u, i \in [\Gamma], u \ \in \ [U]$ such that $k_1^u, k_2^u, \ldots, k_\eta^u > k_{un}$, where $k_{un}$ is defined as $max\{k_{un}^1,k_{un}^2,\ldots,k_{un}^U\}.$ The messages present in each class $i$ are such that, $\mu_i - \ k_i^u \geq k_{un}+1.$ Then the following rate can be achieved.
    \begin{equation}
        R_{PPIR-ISI}= \frac{1}{(k_{un}+1)(\Gamma - \lceil \frac{\eta-1}{U} \rceil)}. \label{multi-user rate}
    \end{equation}
     
\end{theorem}
The achievability scheme in detail is described in the following subsection.

\subsection{Achievability of Theorem 5}
For notational convenience we define the variable $\eta^{\prime} = \lceil \frac{\eta-1}{U} \rceil.$  
Given that $\eta \leq \Gamma$ is the number of identifiable classes, WLOG the first $\eta$ classes are considered to be identifiable and the remaining as unidentifiable for all the users. For the desired class indices denoted by $(v_1,v_2,\ldots,v_U)\text{, where }v_u \in \ [\Gamma]$ and side information set, $S^u, \ \forall \ u \in [U]$, the users send a series of $k_{un}+1$ queries to the server along with the value of $\eta^{\prime}$. The $j^{th}$ query generated by the users is  
$$Q^{(j)} (v,S) = \{ (i,\beta^{(j)}_i)\ \forall \ i \ \in [\Gamma]\}, \ \forall j \in \{1, \ldots,k_{un}+1\},$$ 
where $\beta^{(j)}_i$ is the subclass index. Now for choosing the values of $\beta^{(j)}_i$, the following function and Algorithm 2 is used.\\
\[ \phi {(j)}^{[U]} = \begin{cases} 
          j \bmod U; & j \text{ does not divide } U  \\
          U; & j \text{ divides } U  
       \end{cases}.
    \]

\begin{algorithm}

    \For {$j=1$ \KwTo $k_{un} +1$}{
        \uIf{$1 \leq v_j \leq \eta$ \& $j \leq U$}{
        \vspace{0.2cm}
        $v=v_j$        
        }
        \Else{
        $v=k$\\
        where, $k$ is any arbitrary number between $1$ and $\eta$.
        }
        $u=\phi {(j)}^{[U]}$
        
        $\beta_v ^{(j)} \in [\mu_v]\backslash ( I_{S_v^u} \bigcup\limits_{m=1}^{j-1} \{\beta^{(m)}_v\} )\ \ \ \cdots \cdots \cdots (vii)$\\
        \vspace{0.2cm}
        Choose a subset $Z$ of $[\eta]\backslash \{v\}$ of size $\eta^{\prime}U$ and take its partition as $\{Z_1, Z_2, \ldots , Z_U \}$ such that $|Z_i|= \eta^{\prime}, \ \forall \ i \ \in [U]$.\\

        \For{$i=1$ \KwTo $U$}{
            \For{$t \ \in \ Z_i$}{
               $\beta_{t}^{(j)} \in I_{S_t^i}\backslash  \bigcup\limits_{m=1}^{j-1} \{\beta^{(m)}_t\}  \ \ \ \cdots \cdots  (viii)$  \\
            }
        }
    
        \For{$t \ \in \ [\Gamma]\backslash (Z \cup \{v\})$}{
            $\beta_{t}^{(j)} \in [\mu_t]\backslash  \bigcup\limits_{m=1}^{j-1} \{\beta^{(m)}_t\} $ \\
        }    
    }
    
    \Return{Final result}\;
    \caption{$\beta^{(j)}_i$ for Multi-User PPIR-ISI Scheme}
\end{algorithm}

From Algorithm 2, we can see that if the desired class is unidentifiable for any user $u \in [U]$, then a series of $k_{un} + 1$ queries is generated. In each query, $\eta^{\prime}$ known class-subclass index pairs, i.e., class-subclass index pairs corresponding to the messages available in the side information of each user $u \in [U]$, are sent (refer to eq. (viii) in Algorithm 2). Whereas, if the desired class is identifiable for any user $u \in [U]$, then the user $u$ retrieves a new message from ${u^{th}}$ query itself (refer to eq. (vii) in Algorithm 2) and the remaining $k_{un}+1-U$ queries are sent to maintain the uncertainty between identifiable and unidentifiable classes.

Given, $j \in \{1, 2,\ldots, k_{un}+1\}$, the server generates an answer $A^{(j)}(v,S)$: Encoding of the $\ell$th symbol of messages corresponding to the class-subclass indices $ \{ (i,\beta^{(j)}_i)\ \forall \ i \in [\Gamma]\}$ with a systematic $[2\Gamma - \eta^{\prime} , \Gamma]$ MDS code, $\forall \ \ell \in [L]$ and respond to the user with $(\Gamma - \eta^{\prime})L$ parity symbols.

Consider a query $Q^{(j)}(v, S)$ in which messages corresponding to $\eta^{\prime}$ number of class-subclass index pairs are already known to each user $u \in [U]$ (refer to eq. (viii) in Algorithm 2).
Since, for all $\ell \in [L]$, $(\Gamma - \eta^{\prime})$ parity symbols are sent by the server in answer $A^{(j)}(v,S)$, then from the property of MDS code, each user can obtain all $2\Gamma-\eta^{\prime}$ symbols. The total number of symbols sent by the server is $(k_{un}+1)(\Gamma - \eta^{\prime}) L$.
Thus, we have $$R_{PPIR-ISI}=\frac{L}{(k_{un}+1)(\Gamma - \eta^{\prime}) L} = \frac{1}{(k_{un}+1)(\Gamma - \eta^{\prime})}.$$

\begin{proof}[\textbf{Proof of correctness}] Now we prove that each user will get a new message from its desired class. When the desired class is identifiable for user $u \in [U]$, a series of $k_{un}+1$ queries are generated. Out of these, the $u^{th}$ query is sent with $\eta^{\prime}$ known class-subclass index pair from each Class $i$, $i\in [Z_u]$ (refer to eq. (viii) in Algorithm 2), and an unknown class-subclass index pair from the desired class $v_u$ (refer to eq. (vii) in Algorithm 2). Corresponding to the $u^{th}$ query, the server encodes the $\ell$th symbol of messages, $\forall \ \ell \in [L]$ using a systematic $[2\Gamma - \eta^{\prime}, \Gamma]$ MDS code. Using $(\Gamma - \eta^{\prime})L$ parity symbols and $\eta^{\prime}$ symbols that are known to the User $u$, it can decode the remaining unknown messages in the query. This way, User $u$ will get a new message from the desired class.

When the desired class is unidentifiable for user $u \in [U]$, the user sends $k_{un}+1$ queries to the server, and each query contains $\eta^{\prime}$ class-subclass index pairs (refer to eq. (viii) in Algorithm 2) which are known to the user $u$. For each query, the server encodes the $\ell$th symbol of messages, $\forall \ \ell \in [L]$ using a systematic $[2\Gamma - \eta^{\prime} , \Gamma]$ MDS code. Using $(\Gamma - \eta^{\prime})L$ parity symbols, the user can decode the remaining unknown messages in each query. From Algorithm 2, it is clear that for each class $i \in [\Gamma]$, no class-subclass index pair is repeated in the query that is generated over $k_{un}+1$ times. So at the end, User $u$ is guaranteed to obtain at least one new message from the unidentifiable class, including the desired class, $v_u \in \{\eta+1, \ldots, \Gamma\}$. Hence, the recovery constraint is satisfied.
\end{proof}

\begin{proof}[\textbf{Proof of privacy}] From Algorithm 2, it is clear that for each class $i \in [\Gamma]$, no class-subclass index pair is repeated in all $k_{un}+1$ queries, for any realization of $v_u \in [\Gamma], u \in [U]$. The query $Q^{(j)}(v,S)$ and the answer $A^{(j)}(v,S)$ are thus independent of the desired class $v_u \in [\Gamma]$ $ \forall u \in [U]$. Hence, the privacy constraint is satisfied.

\end{proof}

The following are the justifications for the assumptions in Theorem 5.
\begin{enumerate}
\item Queries are generated over $k_{un}+1$ times, where $k_{un}$ is defined as $\max\{k_{un}^1,\ldots,k_{un}^U\}$. If the desired class of user $u$ is identifiable, a new message from the desired class is recovered in $u^{th}$ query. If the desired class is unidentifiable, then according to Algorithm 2, identifiable classes require at most $k_{un}+1$ messages as its side information (refer to eq. (viii) in Theorem 2). This means that the side information of identifiable classes should be at least $k_{un}+1$. Hence the strict inequality, $k_1^u,k_2^u,\ldots, k_\eta^u > k_{un}, \ \forall \ u \in [U]$.
\item When the desired class of User $u$ is identifiable then it gets satisfied from the answer $A^{(u)}(v, S)$. If $k_{un}+1 < U$, then the demand of $U-k_{un}-1$ users will not be satisfied if their desired class is identifiable. This means that the number of queries generated, i.e., $k_{un}+1$, should be at least $U$.
\item In the proposed scheme, a sufficient number of unknown messages should be present to avoid the repetition of class-subclass indices. Consider the case when the demand class, $v_u$, is the same for all $u \in [U]$. 
Also consider that the ${k_{un}+1}^{th}$ query corresponds to user $u$, where $u \ \in \ [U]$. If the unknowns subclass indices sent in previous $k_{un}$ are the same as the unknown messages to the user $u$ in Class $v_u$, then from $(vii)$ of Algorithm 2, User $u$ should have at least $k_{un}+1$ unknown messages to satisfy the privacy constraint (2). Thus, for all $i \ \in \ [\Gamma]$ and $u \ \in \ [U], \ \mu_i - \ k_i^u \geq k_{un}+1$.
\end{enumerate}

\begin{example}\label{ex5} Consider $U=2$ users connected to a single server with $\Gamma=7$ non-overlapping classes out of which $\eta=5$ classes are identifiable to both the users. The total 53 messages are distributed  in $\Gamma = 7$ classes, each message containing L = 2 symbols from $F_{13}.$
\begin{align*}
C1: &W^{(1)},W^{(8)},W^{(15)},W^{(22)},W^{(29)},W^{(36)},W^{(42)},\\
C2: &W^{(2)},W^{(9)},W^{(16)},W^{(23)},W^{(30)},W^{(37)}, W^{(52)},\\
C3: &W^{(3)},W^{(10)},W^{(17)},W^{(24)},W^{(31)},W^{(50)},W^{(51)},\\
C4: &W^{(4)},W^{(11)},W^{(18)},W^{(25)},W^{(32)},W^{(41)},W^{(43)},W^{(47)}, \\
& W^{(53)},\\
C5:  &W^{(5)},W^{(12)},W^{(19)},W^{(26)},W^{(33)},W^{(38)},W^{(44)},W^{(48)},\\
C6: &W^{(6)},W^{(13)},W^{(20)},W^{(27)},W^{(34)},W^{(39)},W^{(45)},\\
C7: &W^{(7)},W^{(14)},W^{(21)},W^{(28)},W^{(35)},W^{(40)},W^{(46)},W^{(49)}.
\end{align*}
The class-subclass index set, i.e., $\{ (i,\beta_i),\ \forall \ \beta_i\in [\mu_i]\}$ corresponding each message in C1 is 
$\{(1,1),(1,2),(1,3),(1,4),$ $(1,5),(1,6),(1,7)\},$ respectively.
Similarly, a class-subclass index is assigned to each message in C2, C3, C4, C5. Also,
$\mu_1=7,\mu_2=7,\mu_3=7,\mu_4=9,\mu_5=8,\mu_6=7,\mu_7=8$. Let the side information set corresponding to User 1 be $S^1=\{W^{(1)},W^{(8)},W^{(15)},W^{(22)}, W^{(2)}, W^{(9)}, W^{(16)},W^{(23)},W^{(3)},\\ W^{(10)},W^{(17)},W^{(24)},W^{(4)},W^{(11)}, W^{(18)},W^{(25)},W^{(32)},\\ W^{(5)}, W^{(12)},W^{(19)},W^{(26)},W^{(6)},W^{(13)}, W^{(20)}, W^{(7)},  \\ W^{(14)}\}$.
Let the side information set corresponding to User 2 be $S^2=\{ W^{(22)},W^{(29)},  W^{(36)},W^{(42)}, W^{(9)}, W^{(23)}, W^{(30)}, \\ W^{(37)}, W^{(3)},W^{(10)},W^{(24)},W^{(31)},W^{(18)},W^{(25)}, W^{(32)}, \\ W^{(41)}, W^{(43)},W^{(47)},W^{(26)},W^{(33)},W^{(38)},W^{(44)},W^{(20)}, \\ W^{(27)}, W^{(7)},W^{(21)},W^{(28)}\}$.
Observe that, $k_1^1= 4, k_2^1=4, k_3^1=4, k_4^1=5, k_5^1=4, k_6^1=3, k_7^1=2$ and $k_1^2= 4, k_2^2=4, k_3^2=4, k_4^2=6, k_5^2=4, k_6^2=2, k_7^2=3$. WLOG, we take the first five classes as identifiable and the remaining $6^{th} \text{ and } 7^{th}$ as unidentifiable, i.e., the users know the sets $I_{S_1^u}, I_{S_2^u},I_{S_3^u},I_{S_4^u},I_{S_5^u}, \ \forall \ u \in [2]$, and the sets $I_{S_6^u}$ and $I_{S_7^u}$ are not known to the user  $u, \ \forall u \in [2] $. Clearly, $k_{un} = \max\{k_{un}^1,k_{un}^2\} = 3.$
The user sends the value of $\eta^{\prime}=2$ and the following queries to the server.
$$Q^{(j)} (v,S) = \{ (i,\beta^{(j)}_i), \forall \ i \in [7] \}, \  \forall \ j \in \{1,2,3,4\}.$$
\noindent\textit{\textbf{Case 1}}: Let $v_1=2, v_2=3$ be the desired classes.
Using Algorithm 2, the following queries are generated. 
\begin{align*}
Q^{(1)} (v,S) &= \{ (1,1),(2,5),(3,2),(4,3),(5,4),(6,2),(7,8)\},\\
Q^{(2)} (v,S) &= \{ (1,2),(2,1),(3,3),(4,4),(5,6),(6,1),(7,3)\},\\
Q^{(3)} (v,S) &= \{ (1,4),(2,2),(3,1),(4,6),(5,1),(6,4),(7,1)\},\\
Q^{(4)} (v,S) &= \{ (1,3),(2,4),(3,4),(4,1),(5,2),(6,3),(7,5)\}.
\end{align*}
From the above query set, we can see that no subclass index of a particular class is repeated, thus making the server oblivious to the demand class of the user. The desired class, $v_1$ of user 1 is 2. So, in $Q^{(1)} (v, S)$, an unknown subclass index is sent from Class 2 corresponding to User 1. Now consider the set $Z=[\eta]\backslash \{v\}$, where, $v=v_1=2$, i.e., $Z=\{1,3,4,5\}$. According to Algorithm 2, $Z$ is partitioned into $U$ number of $\eta^{\prime}$ sized subsets, for instance, let $Z_1=\{1,3\}$ and $Z_2=\{4,5\}$. For classes, $ i \in \ Z_1$, known subclass indices are sent corresponding to User 1, and for classes, $i \in \ Z_2$, known subclass indices are sent corresponding to User 2. Similarly, for the remaining queries.
For all $ j \in \{1,2,3,4\}$, the server responds to the users with the answer $A^{(j)}(v,S)$:
Encoding of each symbol of the messages corresponding to the class-subclass index pairs $\{ (i,\beta^{(j)}_i), \ \forall \ i \in [7] \}$ with a systematic $[12,7]$ MDS code and respond to the users with $5$ parity symbols. Therefore, the rate is $R=\frac{1}{4 \times 5}=\frac{1}{20}$.
In query $Q^{(1)} (v,S)$, User 1 already knows the messages $W^{(1)}$ and $W^{(10)}$ corresponding to the class-subclass index pairs $(1,1)$ and $(3,2)$, respectively. After receiving $5$ parity symbols, for each $\ell=1,2$, in the answer $A^{(1)}(v,S)$, User 1 will obtain all the messages corresponding to the $Q^{(1)} (v,S)$, i.e., $W^{(1)}, W^{(30)}, W^{(10)}, W^{(18)},W^{(26)},W^{(13)}$ and $W^{(49)}$. The new message obtained from Class 2 is $W^{(18)}$.
Similarly, User 2 will obtain all the messages corresponding to query $Q^{(2)} (v, S)$, and the new messages obtained from Class 3 is $W^{(17)}$.\\
\noindent\textit{\textbf{Case 2}}: Let $v_1=6, v_2=3$ be the desired classes.
Using Algorithm 2, the following queries are generated.
\begin{align*}
Q^{(1)} (v,S) &= \{ (1,1),(2,5),(3,2),(4,3),(5,4),(6,2),(7,8)\},\\
Q^{(2)} (v,S) &= \{ (1,2),(2,1),(3,3),(4,4),(5,6),(6,1),(7,3)\},\\
Q^{(3)} (v,S) &= \{ (1,4),(2,2),(3,1),(4,6),(5,1),(6,4),(7,1)\},\\
Q^{(4)} (v,S) &= \{ (1,3),(2,4),(3,4),(4,1),(5,2),(6,3),(7,5)\}.
\end{align*}
From the above query set, we can see that no subclass index of a particular class is repeated, thus making the server oblivious to the demand class of the users. For all $ j \in \{1,2,3,4\}$, the server responds to the users with the answer $A^{(j)}(v,S)$:
Encoding of each symbol of the messages corresponding to the class-subclass index pairs $\{ (i,\beta^{(j)}_i), \ \forall \ i \in [7] \}$ with a systematic $[12,7]$ MDS Code and respond to the user with $5$ parity symbols. Therefore, the rate is $R=\frac{1}{4 \times 5}=\frac{1}{20}$.
In query $Q^{(j)} (v,S), \forall j \ \in \{1,2,3,4\}$, User 1 already knows at least two messages. Hence, after receiving $5$ parity symbols, for each $\ell=1,2$, the user can decode all the messages in each query, and that includes messages from the unidentifiable classes also. In the end, User 1 will obtain at least one new message from its desired Class 6.
User 2 will obtain all the messages corresponding to query $Q^{(2)} (v,S)$, and the new message obtained from Class 3 is $W^{(17)}$.\\
\noindent\textit{\textbf{Case 3}}: Let $v_1=6, v_2=7$ be the desired classes.
Using Algorithm 2, the following queries are generated.
\begin{align*}
Q^{(1)} (v,S) &= \{ (1,4),(2,5),(3,1),(4,6),(5,1),(6,2),(7,8)\},\\
Q^{(2)} (v,S) &= \{ (1,3),(2,4),(3,4),(4,1),(5,2),(6,1),(7,3)\},\\
Q^{(3)} (v,S) &= \{ (1,5),(2,1),(3,2),(4,3),(5,4),(6,4),(7,1)\},\\
Q^{(4)} (v,S) &= \{ (1,2),(2,2),(3,3),(4,4),(5,6),(6,3),(7,5)\}.
\end{align*}
From the above query set, we can see that no index of a particular class is repeated, thus making the server oblivious to the demand class of the users. For all $ j \in \{1,2,3,4\}$, the server responds to the users with the answer $A^{(j)}(v,S)$: 
Encoding of each symbol of the messages corresponding to the class-subclass index pairs $\{ (i,\beta^{(j)}_i) \ \forall \ i \in [7] \}$ with a systematic $[12,7]$ MDS code and respond to the user with $5$ parity symbols. Therefore, the rate is $R=\frac{1}{4 \times 5}=\frac{1}{20}$. 
In query $Q^{(j)} (v, S), \forall j \ \in \{1,2,3,4\}$, User 1 and User 2 already know at least two messages. Hence, after receiving $5$ parity symbols, for each $\ell=1,2$, the users can decode all the messages in each query, and that includes messages from the unidentifiable classes. In the end, User 1 and User 2 will obtain at least one new message from their desired classes 6 and 7, respectively.
\end{example}

\textbf{Note:} In Example \ref{ex5}, if we run Algorithm 1, $U$ times, then the achieved rate using (5) is $\frac{1}{24}$. Whereas if we use Algorithm 2, then the achieved rate using (\ref{multi-user rate}) is $\frac{1}{20}$.

\section{Conclusion}
In this work, we have presented the problem of single server PPIR-ISI, where, the user is not oblivious to the identity of all side information messages and given a scheme for it. For certain cases we gave analytical proof to show that PPIR-ISI will strictly improve the rate over the PPIR-USI. Further, we gave a scheme for the PPIR-ISI problem with multi-user case.  Some possible directions for further research are:
\begin{itemize}
\item In this paper, Theorems \ref{compare}, \ref{compare2}, and \ref{compare3} present three cases demonstrating that the proposed PPIR-ISI performs better than the PPIR-USI scheme \cite{ref6}. However, beyond these specific cases, there exist other similar scenarios, such as Example 2. Investigating the general conditions under which the proposed PPIR-ISI scheme performs better than the PPIR-USI scheme \cite{ref6} would be interesting.
\item The proposed scheme leverages the identifiability of the classes, under specific conditions. It would be interesting to explore  new PPIR-ISI schemes, which leverage the identifiability of the classes, for scenarios where the given conditions are not satisfied.
\item In this paper, we focus on two possibilities: either the user knows all the subclass indices for its side information belonging to a class, or it knows none of them. However, a more general scenario can be considered where there exist some classes for which the user is aware of the subclass index of certain messages in its side information while remaining unaware of the subclass index for the remaining messages in its side information.
\end{itemize}

\section*{Acknowledgement}
This work was supported partly by the Science and Engineering Research Board (SERB) of the Department of Science and Technology (DST), Government of India, through J.C Bose National Fellowship to Prof. B. Sundar Rajan, by the Ministry of Human Resource Development (MHRD), Government of India, and through Indian Institute of Science (IISc) C V Raman Postdoc Fellowship to Charul Rajput.

\bibliographystyle{ieeetr}
\bibliography{ref}

\end{document}